\documentclass[11pt]{llncs}
\usepackage{amsmath,amssymb}
\usepackage{url}
\usepackage{color}
\usepackage{graphicx}
\usepackage{fullpage}
\usepackage{verbatim}
\usepackage{cite}
\usepackage{hyperref}

\def\N{\mathbb{N}}

\newcommand{\nrank}[1]{\rho_{#1}}
\newcommand{\argmax}{\operatorname{argmax}}

\newcommand{\select}[1]{\operatorname{select}_{#1}}
\newcommand{\rank}[1]{\operatorname{rank}_{#1}}
\newcommand{\abs}[1]{\left|#1\right|}
\newcommand{\menge}[1]{\left\{#1\right\}}
\newcommand{\gauss}[1]{\left\lfloor#1\right\rfloor}
\newcommand{\upgauss}[1]{\left\lceil#1\right\rceil}
\newcommand{\intWort}[1]{\emph{\textbf{#1}}}
\newcommand{\st}{s.t.}

\newcommand{\Wrt}{Wrt.\ }

\newcommand{\tuple}[1]{\left(#1\right)}
\newcommand{\Oh}[1]{\ensuremath{\mathop{}\mathopen{}\mathcal{O}\mathopen{}\left(#1\right)}}
\newcommand{\oh}[1]{\ensuremath{\mathop{}\mathopen{}o\mathopen{}\left(#1\right)}}

\newcommand{\Om}[1]{\ensuremath{\mathop{}\mathopen{}\Omega\mathopen{}\left(#1\right)}}

\newcommand{\bsq}[1]{\lq{#1}\rq}

\newcommand{\instancename}[1]{\mathsf{#1}}

\newcommand{\noderoot}{\instancename{root}}
\newcommand{\LCPA}{\instancename{LCP}}
\newcommand{\ISA}{\instancename{ISA}}
\newcommand{\SA}{\instancename{SA}}
\newcommand{\ST}{\instancename{ST}}
\newcommand{\SucST}{\instancename{SucST}}

\newcommand{\parent}{\instancename{parent}}

\newcommand{\levelanc}{\instancename{level\_anc}}
\newcommand{\strdepth}{\instancename{str\_depth}}

\newcommand{\height}[1]{\ensuremath{\mathop{}\mathopen{}\instancename{height}\mathopen{}\left(#1\right)}}

\newcommand{\citeauthor}[1]{\cite{#1}}

\title{
Lempel Ziv Computation In Small Space (LZ-CISS)
}

\author{
  Johannes Fischer
  \and
  Tomohiro I
  \and
  Dominik K\"{o}ppl
}
\institute{
  Department of Computer Science, TU Dortmund, Germany\\
  \email{\{johannes.fischer, tomohiro.i\}@cs.tu-dortmund.de,}
  \email{dominik.koeppl@tu-dortmund.de}
}

\begin{document}
\maketitle

\begin{abstract}
	For both the Lempel Ziv 77- and 78-factorization we propose algorithms  
	generating the respective factorization using $(1+\epsilon) n \lg n + \Oh{n}$ bits (for any positive constant $\epsilon \le 1$) working space (including the space for the output) for any text of size $n$ over an integer alphabet in $\Oh{n / \epsilon^{2}}$ time.
\end{abstract}

\section{Introduction}
It is difficult to find any practical scenario in computer science for which one could not reason about compression.
Although common focus lies on compression of data on disc storage, for some usages, squeezing transient memory is also practically beneficial.
For instance, the zram module of modern Linux kernels~\cite{RichardIII2014S3} compresses blocks of the main memory in order to prevent the system from running out of working memory.
Compressing RAM is sometimes more preferable than storing transient data on secondary storage (e.g., in a swap file), 
as the latter poses a more severe performance loss.
Another example are websites that usually transferred as ``gzipped'' data by hosting servers~\cite{FGPALZ}.
A server may cache generated webpages in a compressed form in RAM for performance benefits.
To sum up, a common task of these scenarios is the compression and maintenance of data in main memory in order to provide a space-economical, fast access.

Central in many compression algorithms are the LZ77~\cite{Ziv1977uas} or LZ78~\cite{Ziv1978Coi} factorizations.
Both techniques were invented in the late 70's and set a milestone in the field of data compression.
Since main memory sizes of ordinary computers do not scale as fast as the growth of datasets, 
insufficient memory is a well-aware problem; 
both huge mainframes with massive datasets and tiny embedded systems are valid examples for which a simple compressor may end up depleting all RAM.
Besides, they have also been found to be a valuable tool for 
detecting various kinds of regularities in strings~\cite{Crochemore1986TaR,Main1989Dlm,Kolpakov1999FMR,Kolpakov2000FRw,Gusfield2004Lta,Duval2004Lco,Kociumaka2012lta}, 
for indexing~\cite{Kaerkkaeinen1996Lpa,Karkkainen1998LIq,Gagie2012FGS,Gagie2014LSw,Navarro2004Itu,Ferragina2005Ict} and 
for analyzing strings~\cite{Crochemore2003SSA,li05:_lz78_based_strin_kernel,li06:_image_class_via_lz78_based_strin_kernel}.

Large datasets pose a challenge to the main memory budget.
For a solution, one either has to think about algorithm engineering in external memory, or about how to slim down memory consumption during computation in RAM.
\Wrt the latter, we propose an approach that uses $(1+\epsilon) n \lg n + \Oh{n}$ bits (for any positive constant $\epsilon \le 1$) working space 
(including the space for the output) while sustaining linear time computation.
Our approach differs from 
the more recent algorithms (see below),
as it uses a succinct suffix tree representation.

\subsubsection{Related Work. }
While there are naive algorithms that take $\Oh{1}$ working space with quadratic running time (for both LZ77 and LZ78), linear time algorithms with very restricted space emerged only in recent years.

\Wrt LZ77,
the bound of $3n \lg n$ bits set by~\cite{Goto:2013:SFL:2495257.2495853} was very soon lowered to $2n \lg n$ by~\cite{Karkkainen2013LTL}.
For small alphabet size $\sigma$, the upper bound of $n \lg n + \Oh{\sigma \lg n}$ bits by~\cite{Goto2014SEL} is also very compelling.
Their common idea is the usage of previous- and/or next-smaller-value-queries~\cite{cstpp}.
While the approach of K\"arkk\"ainen et al.~\cite{Karkkainen2013LTL} stores SA and NSV completely in two arrays, Goto et al.~\cite{Goto2014SEL} can cope with a single array whose length depends on the alphabet size.
In \cite{Kempa2013LfS}, a practical variant having the worst case performance guarantees of
$(1+\epsilon) n \lg n + n + \Oh{\sigma \lg n}$ bits of working space and $\Oh{n \lg \sigma / \epsilon^2}$ time was proposed.

\Wrt LZ78, 
by using a naive trie implementation, the factorization is computable with $\Oh{z \lg z}$ bits space and $\Oh{n \lg\sigma}$ overall running time, where $z$ is the size of LZ78 factorization. More sophisticated trie implementations~\cite{fischer13alphabet} improve this to $\Oh{n+z{\lg^2\lg\sigma}/{\lg\lg\lg\sigma}}$ time using the same space.

Jansson et al.~\cite{Jansson2015LDT} proposed a compressed dynamic trie based on word packing, and
showed an application to LZ78 trie construction that runs
in $\Oh{n(\lg \sigma +  \lg \lg_{\sigma} n) / \lg_{\sigma} n}$ bits of working space and 
$\Oh{n\lg^2 \lg n / \left(\lg_{\sigma} n \lg \lg \lg n\right)}$ time.
When $\lg \sigma = \oh{{\lg n \lg \lg \lg n}/{\lg^2 \lg n}}$, their algorithm runs even in sub-linear time, but in the worst case it is super-linear.
For an integer alphabet 
a linear time algorithm was recently proposed in~\cite{Nakashima2015CLT},
which utilizes the fact that LZ78 trie is superimposed on the suffix tree of a string.
Although their algorithm works in $\Oh{n \lg n}$ bits of space, they did not care about the constant factor,
and the use of the (complicated) dynamic marked ancestor queries~\cite{amir95improved} seems to 
prevent them from achieving a small constant factor.

\section{Preliminaries}\label{sec:prelim}
Let $\Sigma$ denote an integer alphabet of size $\sigma = \abs{\Sigma} = n^{\Oh{1}}$.
An element $w$ in $\Sigma^*$ is called a \intWort{string},
and $\abs{w}$ denotes its length.
The empty string of length $0$ is called $\varepsilon$.
For any $1 \le i \le \abs{w}$, $w[i]$ denotes the $i$-th character of $w$.
When $w$ is represented by the concatenation of $x, y, z \in \Sigma^*$, i.e., $w = xyz$,
then $x$, $y$ and $z$ are called a \intWort{prefix}, \intWort{substring} and \intWort{suffix} of $w$, respectively.
In particular, a suffix starting at position $i$ of $w$ is called the \intWort{$i$-th suffix} of $w$.
For any $1 \le j \le \abs{w}$, let $S_j(w)$ denote the set of substrings of $w$ that start strictly before $j$.

In the rest of this paper, we take a string $T$ of length $n>0$,
which is subject to LZ77 or LZ78 factorization.
For convenience, let $T[n]$ be a special character that appears nowhere else in $T$,
so that no suffix of $T$ is a prefix of another suffix of $T$. 
Our computational model is the word RAM model with word size $\Om{\lg n}$.
Further, we assume that $T$ is read-only; accessing a word costs $\Oh{1}$ time (e.g., $T$ is stored in RAM using $n \lg \sigma$ bits).

The \intWort{suffix trie} of $T$ is the trie of all suffixes of $T$.
The \intWort{suffix tree} of $T$, denoted by $\ST$, is the tree obtained by compacting the suffix trie of $T$. 
$\ST$ has $n$ leaves and at most $n$ internal nodes.
We denote by $V$ the nodes and by $E$ the edges of $\ST$.
For any edge $e \in E$, the string stored in $e$ is denoted by $c(e)$ and called the \intWort{label} of $e$.
Further, the \intWort{string depth} of a node $v\in V$ is defined as the length of the concatenation of all edge labels on the path from the root to $v$.
The leaf corresponding to the $i$-th suffix is labeled with $i$.
$\SA$ and $\ISA$ denote the suffix array and the inverse suffix array of $T$, respectively~\cite{manber93suffix}.
For any $1 \le i \le n$, $\SA[i]$ is identical to the label of the \emph{lexicographically} $i$-th leaf in $\ST$.
\emph{LCP} and \emph{RMQ} are abbreviations for \emph{longest common prefix} and \emph{range minimum query}, respectively. 
$\LCPA$ is a DS (data structure) on $\SA$ such that $\LCPA[i]$ is the LCP of the \emph{lexicographically} $i$-th smallest suffix with its lexicographic predecessor for $i=2,\ldots,n$.

For any bit vector $B$ with length $\abs{B}$, 
$B.\rank{1}(i)$ counts the number of \bsq{1}-bits in $B[1..i]$, and $B.\select{1}(i)$ gives the position of the $i$-th \bsq{1} in $B$. 
Given $B$, a DS that uses additional $\oh{\abs{B}}$ bits of space and supports
any $\rank{}$/$\select{}$ query on $B$ in constant time can be built in $\Oh{\abs{B}}$ time~\cite{munro01succinct}.

As a running example, we take the string $T = \mathtt{aaabaabaaabaa\$}$.
Since both algorithms for LZ77 and LZ78 are based on the suffix tree, 
we depict the suffix tree of this example string in Fig.~\ref{fig:ST}.

\begin{figure}[h]
\centering{
\includegraphics[scale=0.5]{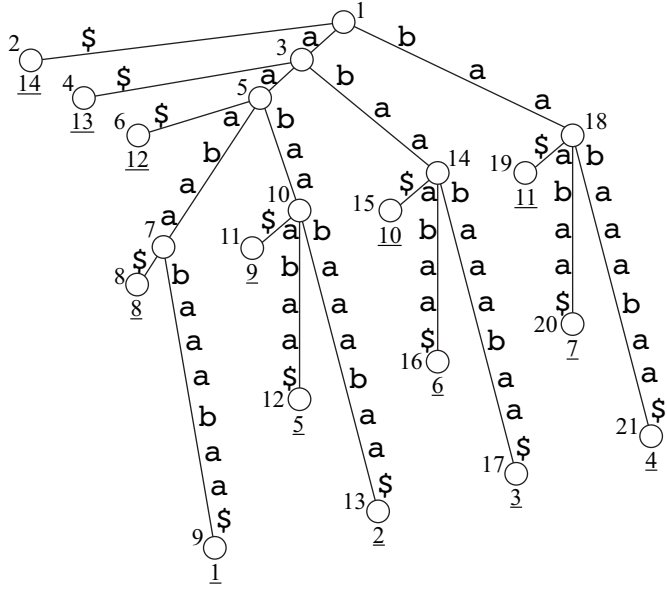}
}
\caption{
  The suffix tree of $T = \mathtt{aaabaabaaabaa\$}$.
  The leaf labels are displayed by the underlined numbers.
  The other numbers show the pre-order of the nodes.
}
\label{fig:ST}
\end{figure}

\subsection{Lempel Ziv Factorization}\label{sec:lz}
A \intWort{factorization} partitions $T$ into $z$ substrings $T=f_1 \cdots f_z$.
These substrings are called \intWort{factors}. In particular, we have:

\begin{definition}
A factorization $f_1\cdots f_z = T$ is called the \intWort{LZ77 factorization} of $T$ iff 
$f_x = \argmax_{S \in S_j(T) \cup \Sigma} \abs{S}$ 
for all $1 \le x \le z$ with $j = \abs{f_1\cdots f_{x-1}}+1$.
\end{definition}
The classic LZ77 factorization adds an additional \intWort{fresh character} to the referencing factors such
that the following definition holds:
\begin{definition} 
A factorization $f_1\cdots f_z = T$ is called the \intWort{classic LZ77 factorization} of $T$ iff 
$f_x$ is the shortest prefix of $f_x \cdots f_z$ that occurs exactly once in $f_1 \cdots f_x$.
\end{definition}

\begin{definition}
A factorization $f_1\cdots f_z = T$ is called the \intWort{LZ78 factorization} of $T$ iff 
$f_x=f'_x\cdot c$ with $f'_x = \argmax_{S \in \menge{f_y : y < x} \cup \menge{\varepsilon} } \abs{S}$ and $c\in\Sigma$ 
for all $1 \le x \le z$.
\end{definition}

We identify factors by text positions, i.e., we call a text position $j$ the \intWort{factor position} of $f_x$ ($1\le x \le z$) iff factor $f_x$ starts at position $j$.
A factor $f_x$ may refer to either (LZ77) a previous text position $j$ (called $f_x$'s \intWort{referred position}), or (LZ78) to a previous factor $f_y$ (called $f_x$'s \intWort{referred factor}---in this case $y$ is also called the \intWort{referred index} of $f_x$).
If there is no suitable reference found for a given factor $f_x$ with factor position $j$, then $f_x$ consists of just the single letter $T[j]$. 
We call such a factor a \intWort{free letter}.
The other factors are called \intWort{referencing factors}.

Our final data structures
allow us to access arbitrary factors (factor position and referred position (LZ77)/referred index (LZ78)) in constant time.

\subsection{Data Structures}\label{sec:cst}

Common to both our algorithms is the construction of a succinct $\ST$ representation.
It consists of $\SA$ with $n \lg n$ bits,
$\LCPA$ with $2n + \oh{n}$ bits,
and a $2|V| + \oh{|V|}$-bit representation of the topology of $\ST$, for which we choose the DFUDS~\cite{Benoit2005RTo} representation.
The latter is denoted by $\SucST$.
We make use of several construction algorithms from the literature:
\begin{itemize}
\item $\SA$ can be constructed in $\Oh{n/\epsilon^2}$ time and $(1+\epsilon)n \lg n$ bits of space, including the space for $\SA$ itself~\cite{karkkainenlinSA}.
\item Given $\SA$, $\LCPA$ can be computed in $\Oh{n}$ time with no extra space~\cite{Valimaki2009Ecs}. Note that $\LCPA$ can only answer $\LCPA[i]$ in constant time 
if $\SA[i]$ is also available.
This is an important remark, because we will discard at several occasions $\SA$ in order to free space, and this discarding causes additional difficulties.
\item Given both $\SA$ and $\LCPA$,
a space economical construction of $\SucST$ was discussed in~\cite[Alg.~1]{cstpp}.
The authors showed that the DFUDS representation of $\ST$ can be built in $\Oh{n}$ time with $n + \oh{n}$ bits of working space.
\end{itemize}

We identify a node $v \in V$ with its pre-order number, which is also the order in which the opening parentheses occur in the DFUDS representation.
So we implicitly identify every node $v \in V$ with its pre-order number (enumerated by $1,\ldots,\abs{V}$).

Since our $\ST$ is static, we can perform various operations on the tree topology in constant time (see, e.g., \cite{cstpp,statdynminmax}).
Among them, we especially use the following operations (for any $v \in V$ and $i \in \N$):
$\parent(v)$ returns the parent of $v$;
and $\levelanc(v, i)$ returns the $i$-th ancestor of $v$.
By building the \emph{min-max tree}~\citeauthor{statdynminmax} on the DFUDS of $\ST$ in $\Oh{n}$ time (using $\Oh{n}$ bits of space),
we can get $\SucST$ supporting these operations in constant time.

Additionally, 
we are interested in answering $\strdepth(v)$ on $\ST$; $\strdepth(v)$ returns the string depth of $v \in V$.
As noted in~\cite{cstpp},
an RMQ data structure on $\LCPA$ can be built in $\Oh{n}$ time and $n + \oh{n}$ bits of working space to support $\strdepth$ in constant time.
Note that the operation $\strdepth$ becomes unavailable when $\SA$ is discarded.

Our algorithms in Sect.~\ref{sec:lz77} and~\ref{sec:lz78} make use of two arrays: 
$A_1$ of size $n\lg n$ bits, and a small helper array $A_2$ of size $\epsilon n \lg n$ bits. 
(We chose such generic names since the contents of these arrays will change several times during the LZ-computation.)

\subsubsection{Node-Marking Vectors. }
In our algorithms, we sometimes deal with subsets $V'$ of $V$.
Pre-order numbers enumerating only the nodes in $V'$ can naturally be used to map nodes in $V'$
to the range $[1..\abs{V'}]$. For this purpose, we use
a \intWort{node-marking vector} $M_{V'}$, which is a bit vector of length $\abs{V}$, such that 
$M_{V'}[v] = 1$ iff $v \in V'$ for any $1 \le v \le \abs{V}$.
We write $\nrank{V'}(v) := M_{V'}.\rank{1}( v )$ for any node $v \in V'$.

\section{LZ77}\label{sec:lz77}

The main idea is to perform leaf-to-top traversals accompanied by the marking of visited nodes.
The marked nodes are indicated by a \bsq{1} in a bit vector of size $\abs{V}$.
Starting from the situation where only the root is marked,
in the $j$-th leaf-to-top traversal for any $1 \le j \le n$, 
we traverse $\ST$ from the leaf labeled with $j$ towards the root, while marking visited nodes
until we encounter an already marked node.
Observe that right before the $j$-th leaf-to-top traversal,
each string of $S_j(T)$ can be obtained by following the path from the root to some marked node.
Hence, the LZ77 factorization can be determined during these leaf-to-top traversals:
If $j$ is a factor position of a factor $f$, the last accessed node $v$ during the $j$-th leaf-to-top traversal reveals $f$'s referred position.
More precisely, $v$ is either the root, or a node that was already marked in a former traversal. 
If $v$ is the root, $f$ is a free letter.
Otherwise, we call $v$ the \intWort{referred node} of $f$.
Then, the factor length is $\strdepth(v)$, 
and the referred position is the minimum leaf label in the subtree rooted at $v$ (retrieved, e.g., by an RMQ on $\SA$).
Since every visited node will be marked, and a marked node will never be unmarked,
the total number of $\parent(\cdot)$-operations is upper bounded by the number of nodes in $\ST$, i.e., $\Oh{n}$.

\subsection{Algorithm}\label{sec:lz77algo}
We start with $\SA$ stored in $A_1[1..n]$, and
some $\Oh{n}$-bit DS to provide $\SucST$, RMQs on $\SA$, and RMQs on $\LCPA$.
Note that the LZ77 computation via leaf-to-top traversals, as explained above,
accesses $\ISA$ $n$ times to fetch suffix leaves that are starting nodes of the traversals,
and accesses $\SA$ $\Oh{z}$ times to compute the factor lengths and the referred positions.
Then, if we have both $\SA$ and $\ISA$,
the LZ77 factorization can be easily done in $\Oh{n}$ time by the leaf-to-top traversals.
However, allowing only $(1+\epsilon)n \lg n + \Oh{n}$ bits for the entire working space, 
it is no longer possible to store both $\SA$ and $\ISA$ completely at the same time.

\subsubsection{With Extra Output Space. } 
Let us first consider the easier case where the result of the factorization can be output \emph{outside} the working space.
We can then use the array+inverse DS of Munro et al.~\cite[Sect.~3.1]{Munro2012Sro},
which allows us to access inverse array's values in $\Oh{1/\epsilon}$ time
by spending additional $\epsilon n \lg n$ bits (on top of the array's size).
Since $\ISA$ is accessed more often than $\SA$,
we first convert $\SA$ on $A_1$ into $\ISA$ and then create its array+inverse DS
so that accessing $\ISA$ and $\SA$ can be done in $\Oh{1}$ and $\Oh{1/\epsilon}$ time, respectively.
Although it is not explicitly mentioned in~\cite{Munro2012Sro}, the DS can be constructed in $\Oh{n}$ time.
Then, the leaf-to-top traversals can be smoothly conducted, leading to $\Oh{z/\epsilon + n}=\Oh{n}$ running time.

Although this is already an improvement over the currently best linear-time algorithm using $2n\lg n$ bits~\cite{Karkkainen2013LTL}, 
doing so would prevent us from also storing the \emph{output} of the LZ77 factorization in the working space.
Solving this is exactly what is explained in the remainder of this section.

\subsubsection{Outline. }
It is difficult to find space for writing the referred positions;
the former algorithm already uses $(1+\epsilon)n \lg n$ bits of working space for the array+inverse DS\@.
Overwriting it would corrupt the DS and cause a problem when accessing $\SA$ or $\ISA$.
We evade this problem by performing several rounds of leaf-to-top traversals during which we build an array that registers every visit of a referred node.
(A minor remark is that this approach does not even need RMQs on $\SA$.) 

Our algorithm is divided into three rounds of leaf-to-top traversals and a final matching phase, all of which will be discussed in detail in the following:

\begin{description}
	\item[First Round:] Construct a bit vector $B_f[1..n]$ marking all factor positions in $T$, and a bit vector $B_r[1..z]$ marking the referencing factors.
             Determine the set of referred nodes $V_r \subset V$, and mark them with a node-marking vector $M_{V_r}$.
	\item[Second Round:] Construct a bit vector $B_D$ counting (in unary) the number of \emph{referred} nodes from $V_r$ visited during each traversal.
	\item[Third Round:] Construct an array $D$ storing the pre-order numbers of all referred nodes visited during each traversal (as counted in the second round).
	\item[Matching:] Convert the pre-order numbers in $D$ to referred positions.
\end{description}
Fig.~\ref{fig:LZ77fact} visualizes the leaf-to-top traversals along with the created data structures $B_D$ and $D$.

\begin{figure}[h]
\centering{
\includegraphics[scale=0.5]{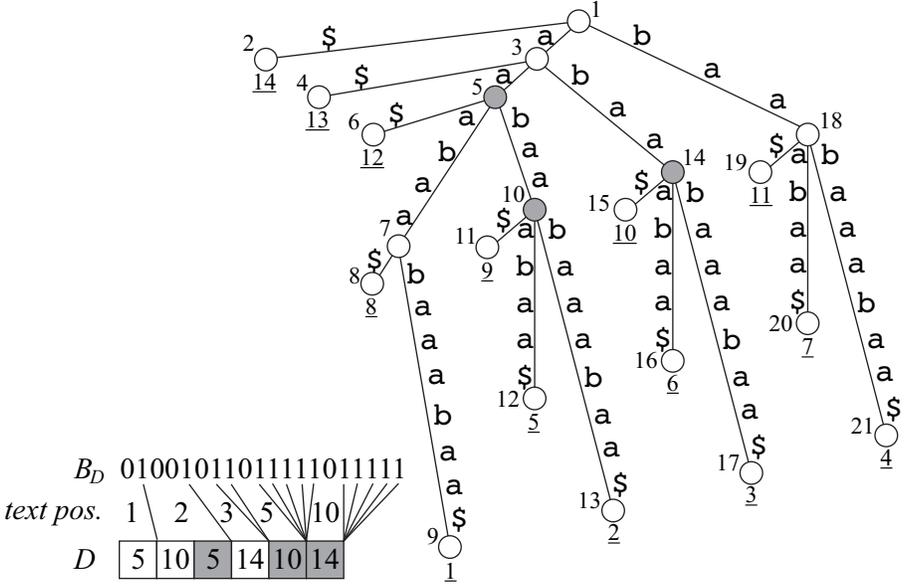}
}
\caption{
The LZ77 factorization partitions $T = \mathtt{aaabaabaaabaa\$}$ as $\mathtt{a|aa|b|aabaa|abaa|\$}$.
  The shaded nodes are the referred nodes. Nodes $5, 10$ and $14$ are referred by $f_2, f_4$ and $f_5$, respectively.
  During the leaf-to-top traversals:
  In the 1st traversal, node $5$ is marked;
  In the 2nd traversal, node $10$ is marked, and node $5$ is referred to by factor $f_2$ with factor position $2$;
  In the 3rd traversal, node $14$ is marked;
  In the 5th traversal, node $10$ is referred to by factor $f_4$ with factor position $5$;
  In the 10th traversal, node $14$ is referred to by factor $f_5$ with factor position $10$;
  Therefore, $B_D = 01001011011111011111$ and
  $D = [5, 10, 5, 14, 10, 14]$, where referred entries are depicted by shaded entries.
}
\label{fig:LZ77fact}
\end{figure}

\subsubsection{Details. }
In the \textbf{first round}, we compute the factor lengths as before by leaf-to-top traversals, which are used to construct $B_f$.
Since the set of referred nodes can be identified during the leaf-to-top traversals, $M_{V_r}$ can be easily constructed.
We also compute $B_r$ by setting $B_r[x] \leftarrow 1$ for every referencing factor $f_x$ with $1 \le x \le z$.
For the rest of the algorithm, the information of $\SA$ is not needed any longer.

We now aim at generating the array $D$ storing a sequence of pre-order numbers of referred nodes,
which will finally enable us to determine the referred positions of each referencing factor.
$D$ is formally defined as a sequence obtained by outputting the pre-orders of referred nodes 
whenever they are marked or referred to during the leaf-to-top traversals.
Hence, each referred node appears in $D$ for the first time when it is marked, 
and after that it occurs whenever it is the last accessed node of the $j$-th traversal, where $1 \le j \le n$ coincides with a factor position. 
To see how $D$ will be useful for obtaining the referred positions,
consider a node $v \in V$ that was marked during the $k$-th traversal.
If we stumble upon $v$ during the $j$-th traversal (for any factor position $j > k$)
we know that $k$ is the referred position for the factor with factor position $j$ (because $v$ had not been marked before the $k$-th traversal).

Alas, just $D$ alone does not tell us \emph{which} referred nodes are found during \emph{which} traversal. 
We want to partition $D$ by the $n$ text positions, 
\st{} we know the traversal numbers which the referred nodes belong to.
This is done by a bit vector $B_D$ that stores a \bsq{1} for each text position $j$, and intersperses these \bsq{1}s with
\bsq{0}s counting the number of referred nodes written to $D$ during the $j$-th traversal.
The size of the $j$-th partition ($1 \le j \le n$) is determined by the number of referred nodes 
accessed during the $j$-th traversal.
Hence the number of \bsq{0}s between the $(j-1)$-th and $j$-th \bsq{1} represents the number of entries in $D$ for the $j$-th suffix.
Formally, $B_D$ is a bit vector 
such that $D[j_b..j_e]$ represents the sequence of referred nodes that are written to $D$ during the $j$-th leaf-to-top traversal, 
where, for any $1 \le j \le n$, $j_b := B_D.\rank{0}(B_D.\select{1}(j-1))+1$ and $j_e := B_D.\rank{0}(B_D.\select{1}(j))$.
Note that for each factor position $j$ of a referencing factor $f$ we encountered its referred node during the $j$-th traversal;
this node is the last accessed node during that traversal, and was stored in $D[j_e]$, which we call the \intWort{referred entry} of $f$.
Note that we do not create a $\rank{0}$ nor a $\select{1}$ DS on $B_D$ because we will get by with sequential scans over $B_D$ and $D$.

Finally, we show the actual computation of $B_D$ and $D$. 
Unfortunately, the computation of $D$ cannot be done in a single round of leaf-to-top traversals; 
overwriting $A_1$ naively with $D$ would result in the loss of necessary information to access the suffix tree's leaves.
This is solved by performing \emph{two} more rounds of leaf-to-top traversals, as already outlined above:
In the \textbf{second round}, with the aid of $M_{V_r}$, $B_D$ is generated by counting the number of referred nodes 
that are accessed during each leaf-to-top traversal.
Next, according to $B_D$, 
we sparsify $\ISA$ by discarding values related to suffixes that 
will not contribute to the construction of $D$ (i.e., those values $i$ for which there is no '0' between the $(i-1)$-th and the $i$-th '1' in $B_D$).
We align the resulting sparse $\ISA$ to the right of $A_1$.
Afterwards, we overwrite $A_1$ with $D$ from left to right in a \textbf{third round} using the sparse $\ISA$. 
The fact that this is possible is proved by the following

\begin{lemma}\label{lemma:Dsize}
  $\abs{D} \le n$.
\end{lemma}
\begin{proof}
First note that the size of $D$ is $\abs{V_r} + z_r$, 
where $z_r$ is the number of referencing factors (number of \bsq{1}s in $B_r$). 
Hence, we need to prove that $\abs{V_r} + z_r \le n$.
  Let $z_{r}^{1}$ (resp.\ $z_{r}^{>1}$) denote the number of referencing factors of length $1$ (resp.\ longer than $1$),
  and let $V_{r}^{1}$ (resp.\ $V_{r}^{>1}$) denote the referred nodes whose string depth is $1$ (resp.\ longer than $1$).
  Also, $z_{f}$ denotes the number of free letters.
  Clearly, $\abs{V_r} = \abs{V_{r}^{1}} + \abs{V_{r}^{>1}}$, $z_r = z_{r}^{1} + z_{r}^{>1}$, $\abs{V_{r}^{1}} \le z_{f}$, and $\abs{V_{r}^{>1}} \le z_{r}^{>1}$.
  Hence $\abs{V_r} + z_r = \abs{V_{r}^{1}} + \abs{V_{r}^{>1}} + z_{r}^{1} + z_{r}^{>1} \le z_{f} + z_{r}^{1} + 2 z_{r}^{>1} \le n$.
  The last inequality follows from the fact that 
  the factors are counted disjointly by $z_{f}$, $z_{r}^{1}$ and $z_{r}^{>1}$, 
  and the sum over the lengths of all factors is bounded by $n$, and every factor counted by $z_{r}^{>1}$ has length at least $2$.
  
\qed\end{proof}

By Lemma~\ref{lemma:Dsize}, $D$ fits in $A_1$.
Since each suffix having an entry in the sparse $\ISA$ has at least one entry in $D$,
overwriting the remaining $\ISA$ values before using them will never happen.

Once we have $D$ on $A_1$, we start {\bf matching} referencing factors with their referred positions.
Recall that each referencing factor has one referred entry, and 
its referred position is obtained by matching the leftmost occurrence of its referred node in $D$.

Let us first consider the easy case with $\abs{V_r} \le \gauss{n\epsilon}$ such that all referred positions fit into $A_2$ (the helper array of size $\epsilon n \lg n$ bits).
By $B_D$ we know the leaf-to-top traversal number (i.e., the leaf's label) during which we wrote $D[i]$ (for any $1 \le i \le \abs{D}$).
For $1 \le m \le \abs{V_r}$, the zero-initialized $A_2[m]$ will be used to store the smallest suffix number at which we found the $m$-th referred node (i.e., the $m$-th node of $V_r$ identified by pre-order).

Let us consider that we have set $A_2[m] = k$, 
i.e., the $m$-th referred node was discovered for the first time by the traversal of the suffix leaf labeled with $k$.

Whenever we read the referred entry $D[i]$ of a factor $f$ with factor position larger than $k$ and 
$\nrank{V_r}(D[i])=m$, we know by $A_2[m]=k$ that the referred position of $f$ is $k$.
Both the filling of $A_2$ and the matching are done in one single, sequential scan over $D$ (stored in $A_1$) from left to right:
While tracking the suffix leaf's label with a counter $1 \le k \le n$, we look at $t := \nrank{V_r}(D[i])$ and $A_2[t]$ for each array position $1 \le i \le \abs{D}$:
if $A_2[t] = 0$, we set $A_2[t] \leftarrow k$.
Otherwise, $D[i]$ is a referred entry of the factor $f$ with factor position $k$, 
for which $A_2[t]$ stores its referred position.
We set $A_1[i] \leftarrow A_2[t]$.
By doing this, we overwrite the referred entry of every referencing factor $f$ in $D$ with the referred position of $f$.

If $\abs{V_r} > \gauss{n\epsilon}$, we run the same scan multiple times, i.e.,
we partition $\menge{1,\ldots,\abs{V_r}}$ into $\upgauss{\abs{V_r}/(n\epsilon)}$ equi-distant intervals (pad the size of the last one) of size $\gauss{n\epsilon}$,
and perform $\upgauss{\abs{V_r}/(n\epsilon)}$ scans.
In order to skip the referred entries in $D$ belonging to an already scanned part of $V_r$, we use a bit vector that marks exactly those positions.
Since each scan takes $\Oh{n}$ time, the whole computation takes $\Oh{\abs{V_r}/\epsilon} = \Oh{z/\epsilon}$ time.

Now we have the complete information of the factorization:
The length of the factors can be obtained by a select-query on $B_f$, and $A_1$ contains the referred positions of all referencing factors.
By a left shift we can restructure $A_1$ such that $A_1[x]$ tells us the referred position (if it exists, according to $B_r[x]$)
for each factor $1 \le x \le z$.
Hence, looking up a factor can be done in $\Oh{1}$ time.

\subsection{Classic LZ77 factorization}
During the leaf-to-top traversals in Section~\ref{sec:lz77algo}, we have to account for the fact that 
the length of each referencing factor has to be enlarged (due to the fresh character).
It suffices to mark the factors in $B_f$ appropriately to the possibly modified lengths ($B_f$ is used to retrieve position and length of any factor); 
the new shape of $B_f$ induces implicitly a modification of $B_r$ and $B_D$.
The fresh character that ends a referencing factor will never be considered to be a factor beginning.
Finally, the fresh character of each referencing factor can be lookup up with $B_f$ and $T$.
Lemma~\ref{lemma:Dsize} still holds for this variant of the factorization; in fact, since $z_{r}^{1} = 0$ and $V_{r}^{1} = \emptyset$,
the proof gets easier.

\section{LZ78} \label{sec:lz78}

Common implementations use a trie for storing the factors.
In the beginning, the trie just consists of the root.
For each newly generated factor we append a leaf to the trie.
If the parent of this leaf is the root, the factor is a free letter,
otherwise it references the factor that corresponds to the parent node.
Hence, each node (except the root) represents a factor. 
We call this trie the \intWort{LZ78 trie}.
Recall that all trie implementations have a (log-)logarithmic dependence on $\sigma$ for top-down-traversals (see the Introduction);
one of our tricks is using $\levelanc$ queries starting from the leaves in order to get rid of this dependence.
For this task we need $\ISA$ to fetch the correct suffix leaf; 
hence, we first overwrite $\SA$ by its inverse.

\subsection{Algorithm}\label{sec:78:outline}
Interestingly, the LZ78 trie is superimposed on the suffix trie of $T$~\cite{bannaiGrammarLz, Nakashima2015CLT}.
Thus, the LZ78 trie structure can be represented by $\ST$,
with an additional DS storing the number of LZ78 trie nodes that lie on each edge of $\ST$.
Each trie node $v$ is called \intWort{explicit} iff it is not discarded during the compactification of the suffix trie towards $\ST$;
the other trie nodes are called \intWort{implicit}.

For every edge $e$ of $\ST$ we use a counting variable $0 \le n_e \le \abs{c(e)}$ that keeps track of how far $e$ is explored.
If $n_e = 0$, then the factorization has not (yet) explored this edge,
whereas $n_e = \abs{c(e)}$ tells us that we have already reached the ending node $v \in V$ of $e =: (u,v)$.
We defer the question how the $n_e$-
and $|c(e)|$-values are stored in $\epsilon n\lg n$ bits to Sect.~\ref{sec:78:bookk}, as those technicalities might not be of interest to the general audience.

Because we want to have a representative node in $\ST$ for \emph{every} LZ78-factor, we introduce the concept of witnesses:
For any $1 \le x \le z$,
the \intWort{witness} of $f_x$ 
is the ST node that is either the explicit representation of $f_x$, 
or, if such an explicit representation does not exist, the ending node in $\ST$ of the edge on which $f_x$ lies.

Our next task is therefore the creation of an array $W[1..z]$ \st{} $W[x]$ stores 
the pre-order number of $f_x$'s witness.
With $W$ it will be easy to find the referred index $y$ of any referencing factor $f_x$.
That is because $f_y$ will either share the witness with $f_x$, or $W[y]$ is the parent node of $W[x]$.
Storing $W$ will be done by overwriting the first $z$ positions of the array $A_1$.

We start by computing $W[x]$ for all $1 \le x \le z$ in increasing order.
Suppose that we have already processed $x-1$ factors, and now want to determine the witness of $f_x$ with factor position $j$.
$\ISA[j]$ tells us where to find the ST leaf labeled with $j$.
Next, we traverse $\ST$ from the root towards this leaf 
(navigated by $\levelanc$ queries in deterministic constant time per edge)
until we find the first edge $e$ with $n_e < \abs{c(e)}$, namely,
$e$ is the edge on which we would insert a new LZ78 trie leaf.
It is obvious that the ending node of $e$ is $f_x$'s witness, which we store in $W[x]$.
We let the LZ78 trie grow by incrementing $n_e$. 
The length of $f_x$ is easily computed by summing up the $\abs{c(\cdot)}$-values along the traversed path, plus $n_e$'s value.
Having processed $f_x$ with factor position $j \in [x..n]$,
$\ISA$'s values in $A_1[1..j]$ are not needed anymore. 
Thus, it is eligible to overwrite $A_1[x]$ by $W[x]$ for $1 \le x \le z$ while computing $f_x$.
Finally, $A_1[1..z]$ stores $W$.
Meanwhile, we have marked the factor positions in a bit vector $B_f[1..n]$. 

For our running example, we conducted the traversals, and marked the witnesses and LZ78 trie nodes superimposed by $\ST$ in Fig.~\ref{fig:LZ78trie}.
\begin{figure}[h]
\centering{
\includegraphics[scale=0.5]{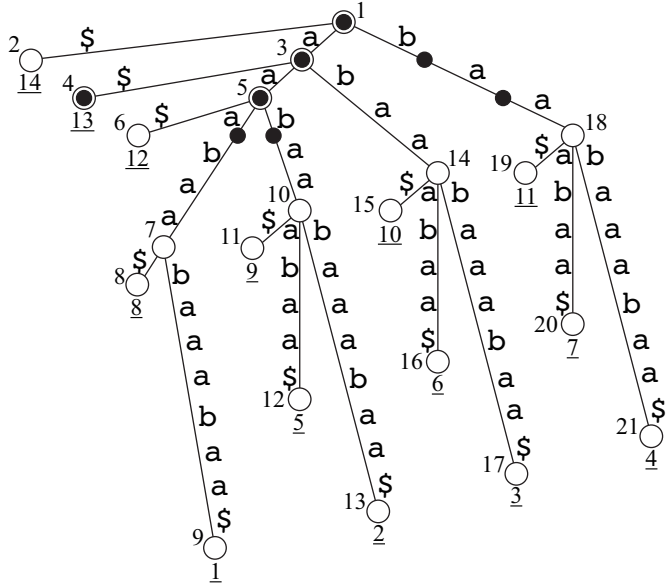}
}
\caption{
  The LZ78 trie for $T = \mathtt{aaabaabaaabaa\$}$ is depicted by bullets on the suffix tree.
  The LZ78 factorization partitions $T$ as $\mathtt{a|aa|b|aab|aaa|ba|a\$}$.
  $W = [3, 5, 18, 10, 7, 18, 4]$ and $V_{\Xi} = \{3, 5, 18\}$.
}
\label{fig:LZ78trie}
\end{figure}

Matching the factors with their references can now be done in a top-down-manner by using $W$.
Let us consider a referencing factor $f_x$ with referred factor $f_y$.
We have two cases:
Whenever $f_y$ is explicitly represented by a node $v$ (i.e., by $f_y$'s witness), 
$v$ is the parent of $f_x$'s witness.
Otherwise, $f_y$ has an implicit representation and hence has the same witness as $f_x$.
Hence, if
$W$ stores at position $x$ the \emph{first} occurrence of $W[x]$ in $W$,
$f_y$ is determined by the largest position $y<x$ for which $W[y]=\parent(v)$;
otherwise ($W[x]$ is \emph{not} the first occurrence of $W[x]$ in $W$),
then the referred factor of $f_x$ is determined by the largest $y<x$ with $W[x]=W[y]$.

Now we hold $W$ in $A_1[1..z]$, leaving us $A_1[z+1..n]$ as free working space that will be used to store a new array $R$,
storing for each witness $w$ the index of the most recently processed factor whose witness is $w$.
However, reserving space in $R$ for \emph{every} witness would be too much (there are potentially $z$ many of them); 
we will therefore have to restrict ourselves to a carefully chosen subset of witnesses. This is explained next.

First, let us consider a witness $w$ that is witnessed by a single factor $f_x$ whose LZ78 trie node is a leaf. 
Because no other factor will refer to $f_x$, we do not have to involve $w$ in the matching.
Therefore, we can neglect all such witnesses during the matching.
The other witnesses (i.e., those being witnessed by at least one factor that is not an LZ78 trie leaf) 
are collected in a set $V_\Xi$ and marked by a bit vector $M_{V_\Xi}$.
$\abs{V_\Xi}$ is at most the number $z_i$ of internal nodes of the LZ78 trie, which is bounded by $n-z$, due to the following

\begin{lemma}\label{lemma:witnessSpace}
$z + z_i \le n$.
\end{lemma}
\begin{proof}
Let $\alpha$ (resp. $\beta$) be the number of free letters that are internal LZ78 trie nodes (resp. LZ78 trie leaves).
Also, let $\gamma$ (resp. $\delta$) be the number of referencing factors that are internal LZ78 trie nodes (resp. LZ78 trie leaves).
Obviously, $\alpha+\beta+\gamma+\delta = z$.
\Wrt the factor length, each referencing factor has length of at least $2$,
while each free letter is exactly one character long.
Hence $2(\gamma+\delta)+\alpha+\beta = z + \gamma + \delta \le n$.
Since each LZ78 leaf that is counted by $\delta$ has an LZ78 internal node of depth one as ancestor (counted by $\alpha$),
$\alpha \le \delta$ holds.
Hence, $z+z_i \le z+\alpha+\gamma \le z+\gamma+\delta \le n$.
\qed\end{proof}

By Lemma~\ref{lemma:witnessSpace}, if we let $R$ store only the indices of factors whose witnesses are in $V_\Xi$, it fits into $A_1[z+1..n]$, and
we can use $M_{V_\Xi}$ to address $R$.

We now describe how to convert $W$ (stored in $A_1[1..z]$) into the referred indices, such that in the end $A_1[x]$ contains the referred index of $f_x$ for $1 \le x \le z$.
We scan $W = A_1[1..z]$ from left to right while keeping track of
the index of the most recently visited factor that witnesses $v$,
for each witness $v \in V_\Xi$ at $R[\nrank{V_\Xi}(v)]$.
Suppose that we are now processing $f_x$ with witness $v = W[x]$.
\begin{itemize}
\item If $v \notin V_\Xi$ or $R[\nrank{V_{\Xi}}(v)]$ is empty,
      we are currently processing the first factor that witnesses $v$.
      Further, if $f_x$ is not a free letter,
      its referred factor is explicitly represented by the parent of $v$.
      We can find its referred index at position $\nrank{V_\Xi}( \parent(v) )$ in $R$.
\item Otherwise, $v \in V_\Xi$, and $R[\nrank{V_\Xi}(v)]$ has already stored a factor index. Then
      $R[\nrank{V_\Xi}(v)]$ is the referred index of $f_x$.
\end{itemize}
In either case, if $v \in V_\Xi$, we update $R$ by writing the current factor index $x$ to $R[\nrank{V_{\Xi}}(v)]$.
Note that after processing $f_x$, the value $A_1[x]$ is not used anymore.
Hence we can write the referred index of $f_x$ to $A_1[x]$ (if it is a referring factor) or set $A_1[x] \leftarrow 0$ (if it is a free letter).
In the end, $A_1[1..z]$ stores the referred indices of every referring factor.

Now we have the complete information about the LZ78 factorization:
For any $1 \le x \le z$, $f_x$ is formed by $f_y c$, where $y = A_1[x]$ is the referred index and $c = T[B_f.\select{1}(x+1)-1]$ the additional letter (free letters will refer to $f_0 := \varepsilon$).
Hence, looking up a factor can be done in $\Oh{1}$ time.

\subsection{Bookkeeping the LZ78 Trie Representation}\label{sec:78:bookk}

Basically, we store both $n_e$ and $\abs{c(e)}$ for each edge $e$ so as to represent the LZ78 trie construction in each step.
A naive approach would spend $2 \lg (\max_{e \in E} \abs{c(e)})$ bits for every edge, i.e., $4 n \lg n$ bits in the worst case.
In order to reduce the space consumption to $\epsilon n \lg n + \oh{n}$ bits, 
we will exploit two facts: 
(1) the superimposition of the LZ78 trie on $\ST$ takes place only in the \emph{upper} part of $\ST$, 
and (2) most of the needed $\abs{c(e)}$- and $n_e$-values are actually small.

More precisely, we will introduce an upper bound for the $n_e$ values, 
which shows that the necessary memory usage for managing the $n_e$ and $\abs{c(e)}$ values is, without a priori knowledge of the LZ78 trie's shape, actually very low.

Note that although we do not know the LZ78 trie's shape, we will reason about those nodes that might be created by the factorization.
For a node $v \in V$, let $\height{v}$ denote the height of $v$ in the LZ78 trie if $v$ is the explicit representation of an LZ78 trie node; 
otherwise we set $\height{v} = 0$.

For any node $v \in V$, let $l(v)$ denote the number of descendant leaves of $v$.
The following lemma gives us a clue on how to find an appropriate upper bound:
\begin{lemma}\label{lemma:lz78height}
Let $u,v \in V$ with $e := (u,v) \in E$. Further assume that $u$ is the explicit representation of an LZ78 trie node.
Then $\height{v}$ is upper bounded by $l(v) - \abs{c(e)}$.
\end{lemma}
\begin{proof}
Let $\pi$ be a longest path from $u$ to some descendant leaf of $v$, and
$d:=\height{v}+\abs{c(e)}$ (i.e., the number of LZ78 trie edges along $\pi$).
 By construction of the LZ78 trie, the ST node $v$ must have at least $d$ leaves,
 for otherwise the (explicit or implicit) LZ78 trie nodes on $\pi$
will never get explored by the factorization.
  So $d \le l(v)$, and the statement holds.
\qed\end{proof}
Further, let $\noderoot$ denote the root node of the suffix \emph{trie}. 
In particular, $\noderoot$ is an explicit LZ78 trie node.
Consider two arbitrary nodes $u, v \in V$ with $e:=(u,v) \in E$.
Obviously, the suffix \emph{trie} node of $v$ is deeper than the suffix \emph{trie} node of $u$ by $\abs{c(e)}$.
Putting this observation together with Lemma~\ref{lemma:lz78height}, we define $h : V \rightarrow \N_0$, 
which upper bounds $\height{\cdot}$:
\[ 
	h(v) = 
	\begin{cases}
		n & \text{~if~} v = \noderoot, \\
		\max \tuple{0, \min \tuple{h(u), l(v)} - \abs{c(e)}} & \text{~if there is an~} e := (u,v) \in E. 
	\end{cases}
\]
Since the number of LZ78 trie nodes on an edge below any $v \in V$ is a lower bound for $\height{v}$, we conclude with the following lemma:

\begin{lemma}\label{lemma:ub}
For any edge $e = (v, w) \in E$, $n_e \le \min \tuple{\abs{c(e)}, h(v)}$.
\end{lemma}
Let us remark that Lemma~\ref{lemma:ub} does not yield a tight bound.
For example, the height of the LZ78 trie is indeed bounded by $\sqrt{2n}$ (see, e.g., \cite[Lemma~1]{bannaiGrammarLz}).
But we do not use this property to keep the analysis simple.

Instead, we classify the edges $e \in E$ into two sets, depending on whether $n_e \le \Delta := \gauss{n^{\epsilon/4}}$ holds for sure or not.
By Lemma~\ref{lemma:ub},
this classification separates $E$ into 
$E_{\le\Delta} := \{ (u,v) \in E : \min \tuple{\abs{c((u,v))}, h(u)} \le \Delta \}$ and $E_{>\Delta} := E \setminus E_{\le\Delta}$.
Since $2 \lg \Delta$ bits are enough for bookkeeping any edge $e \in E_{\le\Delta}$, 
the space needed for these edges fits in $2 \abs{E_{\le\Delta}} \lg \Delta \le n \epsilon \lg n$ bits.
Thus, our focus lies now on the edges in $E_{>\Delta}$; each of them costs us $2 \lg n$ bits.
Fortunately, we will show that $\abs{E_{>\Delta}}$ is so small that
the space of $2 \abs{E_{>\Delta}} \lg n$ bits needed by these edges is in fact $\oh{n}$ bits.

We call any $e \in E_{>\Delta}$ a \intWort{$\Delta$-edge} and its ending node a \intWort{$\Delta$-node}.
The set of all $\Delta$-nodes is denoted by $V_\Delta$.
As a first task, let us estimate the number of $\Delta$-edges
on a path from a node $v \in V_\Delta$ to any of its descendant leaves;
because $v$ is a $\Delta$-node with $\height{v} \le h(v)$, this number is upper bounded by
$\gauss{\frac{h(v)}{\Delta}} \le \gauss{\frac{l(v) - \Delta}{\Delta}} = \gauss{\frac{l(v)}{\Delta}} - 1$.
For the purpose of analysis, 
we introduce $\hat{h} : (V_\Delta \cup \menge{\noderoot}) \rightarrow \N_0$, which upper bounds the number of 
$\Delta$-edges that occur on a path from a node to any of its descendant leaves:
\[ 
	\hat{h}(v) = 
	\begin{cases}
		\gauss{\frac{n}{\Delta}} & \text{~if~} v = \noderoot, \\
		\min \tuple{\hat{h}(\hat{p}(v)) - 1, \gauss{\frac{l(v)}{\Delta}} - 1} & \text{~otherwise},
	\end{cases}
\]
where $\hat{p} : V_\Delta \rightarrow (V_\Delta \cup \menge{\noderoot})$ 
returns for a node $v$ either its deepest ancestor that is a $\Delta$-node, 
or the root if such an ancestor does not exist.
Note that $\hat{h}$ is non-negative by the definition of $V_\Delta$.

For the actual analysis, $\alpha(v)$ shall count 
the number of $\Delta$-edges in the subtree rooted at $v \in V_\Delta \cup \menge{\noderoot}$.
\begin{lemma}\label{lemma:DeltaLessOcc}
For any node $v \in (V_\Delta \cup \menge{\noderoot})$,
	$
		\alpha(v) \le \frac{l(v)}{\Delta} \sum_{i=1}^{\hat{h}(v)} \frac{1}{i}.
	$
\end{lemma}

\begin{proof}
We proceed by induction over the values of $\hat{h}(v)$ for every $v\in V_\Delta$.
For $\hat{h}(v) = 0$ the subtree rooted at $v$ has no $\Delta$-edges; hence $\alpha(v) = 0$.
If $\hat{h}(v) = 1$, any $\Delta$-node $w$ of the subtree rooted at $v$ holds the property $\hat{h}(w) = 0$.
Hence, none of those $\Delta$-nodes are in ancestor-descendant relationship to each other.
By the definition of $\Delta$-nodes, for any $\Delta$-node $u$, 
we have $0 \le \gauss{\frac{l(u)}{\Delta}} - 1$, and hence, $\Delta \le l(u)$.
By $\Delta \alpha(v) \le \sum_{u \in V_{\Delta}, \hat{p}(u) = v} l(u) \le l(v)$ we get
$\alpha(v) \le \frac{l(v)}{\Delta}$.

For the induction step,
let us assume that the induction hypothesis holds for every $u\in V_\Delta$ with $\hat{h}(u) < k$.
Let us take a $v \in V_\Delta$ with $\hat{h}(v) = k$.
Further, let
$V_{k'} := \menge{ u \in V_\Delta : \hat{p}(u) = v \text{~and~} \hat{h}(u) = k'} $ for $0 \le k' \le k-1$
denote the set of $\Delta$-nodes that have the same $\hat{h}$ value and are descendants of $v$, without having a $\Delta$-node as ancestor that is a descendant of $v$.
These constraints ensure that there does not exist any $u \in \bigcup_{0 \le k' \le k-1} V_{k'} =: \mathcal{V}$  
that is ancestor or descendant of some node of $\mathcal{V}$. 
Thus the sets of descendant leaves of the nodes of $\mathcal{V}$ are disjoint.
So it is eligible to denote by $L_{k'} := \sum_{u \in V_{k'}} l(u)$ the number of descendant leaves of all nodes of $V_{k'}$.
It is easy to see that $\sum_{k'=0}^{k-1} L_{k'} \le l(v)$.
Now, by the hypothesis, and the fact that each $u \in \mathcal{V}$ is the highest $\Delta$-node on every path from $v$ to any leaf below $u$, we get
\[
	\alpha(v) 
	\le 
	\abs{V_{0}} + \sum_{k'=1}^{k-1} \left( \sum_{u\in V_{k'}} \frac{l(u)}{\Delta} \sum_{i=1}^{k'} \frac{1}{i} 
	+ \abs{V_{k'}} \right)
	=
	\abs{V_{0}} + \sum_{k'=1}^{k-1} \left( \frac{L_{k'}}{\Delta} \sum_{i=1}^{k'} \frac{1}{i} + \abs{V_{k'}} \right).
\]
By definition of $V_{k'}$ and $\hat{h}$, we have $\hat{h}(u) = k' \le \gauss{\frac{l(u)}{\Delta}}-1$ and hence
$(k'+1)\Delta \le l(u)$ for any $u \in V_{k'}$.
This gives us $\frac{L_{k'}}{(k'+1)\Delta} = \sum_{u \in V_{k'}} \frac{l(u)}{(k'+1)\Delta} \ge \abs{V_{k'}}$.
In sum, we get
\[
	\alpha(v) \le \frac{L_{0}}{\Delta} + \sum_{k'=1}^{k-1} \frac{L_{k'}}{\Delta} \sum_{i=1}^{k'+1} \frac{1}{i} 
        = \sum_{k'=0}^{k-1} \frac{L_{k'}}{\Delta} \sum_{i=1}^{k'+1} \frac{1}{i} 
	\le \frac{l(v)}{\Delta} \sum_{i=1}^k \frac{1}{i}\ .
\]
\qed\end{proof}

By Lemma~\ref{lemma:DeltaLessOcc}, $\abs{E_{>\Delta}} = \alpha(\noderoot) \le \frac{n}{\Delta} \sum_{i=1}^{\frac{n}{\Delta}} \frac{1}{i}$.
Since
$\sum_{i=1}^{\frac{n}{\Delta}} \frac{1}{i} \le 1 + \ln\frac{n}{\Delta}$,
 we have $\alpha(\noderoot) \le \frac{n}{\Delta} + \frac{n}{\Delta} \ln\frac{n}{\Delta} = \Oh{\frac{n}{\Delta} \lg\frac{n}{\Delta}} 
= \Oh{n\lg n / \left(n^{\epsilon/4 }\right)}$.
We conclude that the space needed for $E_{>\Delta}$ is $2 \abs{E_{>\Delta}} \lg n = \Oh{\frac{n\lg^2 n}{n^{\epsilon/4}}} = \oh{n}$ bits.

Finally, we explain how to implement the data structures for bookkeeping the LZ78 trie representation.
By an additional node-marking vector $M_{V_\Delta}$ that marks the $V_\Delta$-nodes, we divide the edges into $E_{\le\Delta}$ and $E_{>\Delta}$.
$\rank{}/\select{}$ on $M_{V_\Delta}$ allows us to easily store, access and increment the $n_e$ values for all edges in constant time.
$M_{V_\Delta}$ can be computed in $\Oh{n}$ time when we have $\SA$ on $A_1$:
since $\strdepth$ allows us to compute every $\abs{c(e)}$ value in constant time,
we can traverse $\ST$ in a DFS manner while computing $h(v)$ for each node $v$, and hence,
it is easy to judge whether the current edge belongs to $E_{>\Delta}$.
In order to store the $h$ values for all ancestors of the current node we use a stack.
Observe that the $h$ values on the stack are monotonically increasing;
hence we can implement it using a DS with $\Oh{n}$ bits~\cite[Sect.~4.2]{Fischer2010OSR}.

\bibliographystyle{splncs03}
\bibliography{ref,literatur}

\end{document}